\documentclass[final,onecolumn]{IEEEtran}

\usepackage{cite}
\usepackage{amsmath,amssymb,amsfonts}
\usepackage{bm}
\usepackage{algorithmic}
\usepackage{graphicx}
\usepackage{textcomp}
\usepackage{xcolor}
\usepackage{rotating, booktabs}
\usepackage{wallpaper}
\usepackage{amsmath}
\usepackage{amsthm}
\usepackage{amsfonts}
\usepackage{amssymb}
\usepackage{graphicx}
\usepackage{float}
\usepackage{subfigure}
\usepackage{makecell}
\usepackage{hyperref}
\usepackage{indentfirst}
\usepackage{comment}

\newtheorem{definitionenv}{Definition}
\newtheorem{lemmaenv}[definitionenv]{Lemma}
\newtheorem{theoremenv}[definitionenv]{Theorem}
\newtheorem{corollaryenv}[definitionenv]{Corollary}
\newtheorem{propositionenv}[definitionenv]{Proposition}
\newtheorem{conjectureenv}[definitionenv]{Conjecture}
\newtheorem{remarkenv}[definitionenv]{Remark}
\newenvironment{remark}{\begin{remarkenv}\rm}{\end{remarkenv}}
\newcommand{\er}{\end{remark}}

\newtheorem{exampleenv}{Example}
\newtheorem{app-lemmaenv}[section]{Lemma}

\newenvironment{definition}{\begin{definitionenv}\rm}{\end{definitionenv}}
\newenvironment{lemma}{\begin{lemmaenv}\rm}{\end{lemmaenv}}
\newenvironment{theorem}{\begin{theoremenv}\rm}{\end{theoremenv}}
\newenvironment{corollary}{\begin{corollaryenv}\rm}{\end{corollaryenv}}

\newenvironment{proposition}{\begin{propositionenv}\rm}{\end{propositionenv}}

\newenvironment{app-lemma}{\begin{app-lemmaenv}\rm}{\end{app-lemmaenv}}

\theoremstyle{definition}

\newcommand{\cA}{{\mathcal A}}

\newcommand{\cP}{{\mathcal P}}

\newcommand{\bmo}{{\bm 0}}
\newcommand{\bml}{{\bm 1}}

\newcommand{\bma}{{\bm a}}
\newcommand{\bmb}{{\bm b}}
\newcommand{\bmc}{{\bm c}}

\newcommand{\bmi}{{\bm i}}
\newcommand{\bmj}{{\bm j}}
\newcommand{\bmn}{{\bm n}}
\newcommand{\bmk}{{\bm k}}

\newcommand{\bms}{{\bm s}}
\newcommand{\bmt}{{\bm t}}

\newcommand{\bmp}{{\bm p}}

\begin{document}

\title{Semidefinite programming bounds for binary codes from a split Terwilliger algebra}
\author{Pin-Chieh Tseng, Ching-Yi Lai, and Wei-Hsuan Yu
		
\thanks{\footnotesize This article was presented in part at ISIT 2022 \cite{PLY22}.
Pin-Chieh Tseng and Ching-Yi Lai are with the Institute of Communications Engineering, National Yang Ming Chiao Tung University (NYCU), Hsinchu 30010, Taiwan. Pin-Chieh Tseng is also with the Department of Applied Mathematics, NYCU. (emails: pichtseng@gmail.com and cylai@nycu.edu.tw)
			
Wei-Hsuan Yu is with the Department of Mathematics, National Central University, Taoyuan 32001, Taiwan. (email: u690604@gmail.com)
			
}}
	
\date{\today}

\maketitle
	
\begin{abstract}
We study the upper bounds for $A(n,d)$, the maximum size of codewords with length $n$ and Hamming distance at least $d$. Schrijver studied the Terwilliger algebra of the Hamming scheme and proposed a semidefinite program to bound $A(n, d)$. We derive more sophisticated matrix inequalities based on a split Terwilliger algebra to improve Schrijver's semidefinite programming bounds on $A(n, d)$. In particular, we improve the semidefinite programming bounds on $A(18,4)$ to $6551$.
\end{abstract}


\section{Introduction}
In coding theory, one of the classical problems is to determine $A(n, d)$, the maximum size of a binary $(n,d)$ code with length $n$ and minimum distance at least $d$. The (Hamming) distance distribution of a code is considered since  the MacWilliams identities provide a linear relation between the distance distribution and its transform~\cite{MS77}. With this linear relation, the maximum possible size of a code can be formulated as  {a linear programming problem}. Delsarte showed that the transform of the distance distribution is nonnegative and used linear programming techniques to derive upper bounds on  $A(n, d)$~\cite{Del73}.
 
Schrijver considered the distance relations among triplets of codewords and derived positive-semidefinite relations based on the Terwilliger algebra of Hamming scheme \cite{Sch05}. Thus one can formulate a semidefinite program (SDP) on the maximum possible size of a code and derive a semidefinite programming bound on the size of a binary code. Several linear programming upper bounds on $A(n,d)$ are improved since Schrijver's semidefinite constraints imply Delsarte's linear inequalities~\cite{Sch05} by diagonalizing certain positive semidefinite matrices derived from Schrijver's semidefinite constraints in the Bose-Mesner algebra.

This method was later extended  to nonbinary codes by Gijswijt, Schrijver, and Tanaka \cite{GST06}. Gijswijt, Mittelmann, and Schrijver further studied  the distance relations among quadruples of codewords and generalized Schrijver's SDP \cite{GMS12}, called quadruple SDP, which improved many upper bounds for $A(n,d)$. Although an SDP based on the distance relations of $m$-tuple codewords for any $m$ is studied, an SDP based on quadruple distances has already many variables, leading to high computation complexity.
 
On the other hand, there are several known linear constraints for binary codes, including Delsarte's  inequalities~\cite{Del73},   the ones derived by Best~\cite{Bes80} and Mounits, Etzion and Litsyn \cite{MEL02}, which can be used to strengthen linear programming or semidefinite programming bounds on $A(n,d)$. Moreover, Kim and Toan proved additional linear constraints on the variables of Schrijver's SDP  and improved upper bounds on $A(18, 8)$ and $A(19, 8)$ \cite{KT13}. Then, $A(18, 8) = 64$ has later been settled by {\"O}sterg{\aa}rd \cite{Ost19} by using a computer-aided search.

The $A(n,d)$ problem can be regarded as finding the maximum number of an independent set of a graph as follows. Let $E$ be a graph with $2^n$ vertices corresponding to all the binary vectors of length $n$. There is an edge between two binary vectors if their Hamming distance is less than $d$. Now an $(n,d)$ code corresponds to an independent set of $E$. Consequently, Delsarte's linear programming bound can be viewed an upper bound on the independent number of $E$. Moreover, this bound can be extended to serve as an upper bound on the independent number an arbitrary graph \cite{Sch79}. Based on this connection, Laurent  gave a hierarchy for semidefinite programming bounds on $A(n, d)$  and  proposed strengthened bounds~\cite{Lau07}, which improve bounds on $A(20, 8)$ and $A(25, 6)$. 

Upper bounds for several related coding problems in various spaces can also be derived using semidefinite programming techniques. For instance, Bachoc and Vallentin  studied SDPs for codes in Hamming balls, projective spaces and spherical codes (kissing number problems) \cite{Bac10,BV08}. Barg and Yu also used semidefinite programming techniques to obtain better upper bounds for spherical two-distance sets and equiangular lines~\cite{BY13, BY14}.
 
In this paper, we would like to study Schrijver's SDP and derive additional semidefinite constraints. One can define a split distance distribution of a code, and  derive a split version of Delsarte's inequalities, which provide subtler linear constraints~\cite{Sim95}.
Recently, split Hamming weight distributions and their MacWilliams identities have been studied in various quantum codes~\cite{LHL16,LA18,ALB20}. It has been demonstrated that linear programming bounds on quantum codes can be improved with additional constraints from  split MacWilliams identities~\cite{LA18}. 
 
Inspired by the effects of split distance or weight distributions in linear programming, we would like to study a similar notion in Schrijver's SDP. Consider a partition of the support of a code with two subsets. We define a split Terwilliger algebra with respect to the partition. Similar to the derivation of Schrijver's semidefinite constraints, we show that this split Terwilliger algebra can be {block}-diagonalized to derive finer positive-semidefinite constraints. Moreover, we show that these split Schrijver's semidefinite constraints also imply corresponding split Delsarte's inequalities, and hence they are natural generalizations of Schrijver's  constraints. Together with Schrijver's semidefinite constraints and the known linear constraints in the literature, we have a strengthened SDP on $A(n,d)$. In particular, we improve the semidefinite programming bound on $A(18,4)$ to $6551$, while the previously known upper bound is $6552$, by linear programming with Delsarte's inequalities and Best's inequalities~\cite{Bes80}. The number has not been updated since more than four decades ago~\cite{BBMOS_1987,Brouwer}. The numerical error of this SDP program can be pessimistically estimated from its dual SDP as suggested by Gijswijt~\cite{Gij05}. Using this method, we are able to verify that $A(18,4) \leq 6551$.

One of our semidefinite constraint $R_{s}$ can be derived from the quadruple SDP. It is not clear whether the other semidefinite constraint on $R_{s}'$ obtained by a split Terwilliger algebra is included in the constraints of the quadruple SDP as well. However, in our experiment, we are able to improve the bound $A(18,4)$ over the quadruple SDP in \cite{GMS12}.

All the results and proofs can be generalized to a split Terwilliger algebra on $m$ subsets of a partition with  $m \leq n$ (called \textit{$m$-split Terwilliger algebra}), from which we may derive more additional positive-semidefinite constraints.  
However, this $m$-split Terwilliger induces an SDP with  $O((\frac{n}{m})^{3m})$ variables and may not be practical in implementation with large $m$.

Finally, we mention $m$-split Terwilliger algebras for the Hamming scheme, which might allow us to apply our method to other association schemes. To implement an SDP program, one of the key points is to block diagonalize the algebra in use
and  Gijswijt has developed a general method to handle this problem~\cite{Gij09}. Gijswijt's method was refined and extended to nonbinary codes by Litjens, Polak and Schrijver \cite{LPS17}. Moreover, the method can be further generalized to the groups of the form $(G^{n_{1}} \rtimes S_{n_{1}}) \times \cdots \times (G^{n_{m}} \rtimes S_{n_{m}})$ with $\sum_{i = 1}^{m}n_{m} = n$ and $S_{n_{i}}$ are symmetric groups \cite{Pol19_1}. Together with those approaches, one can calculate the block diagonalization formula of split Terwilliger algebras in various types of codes, such as constant-weight codes and nonbinary codes with Hamming or Lee distances.

The paper is organized as follows. We introduce the Terwilliger algebra  of the Hamming scheme and Schrijver’s SDP. In Section~\ref{sec:split} we define a split Terwilliger algebra and derive semidefinite constraints. Then we provide our SDP together with the linear constraints in the literature in Section~\ref{sec:SDP}. A generalization of the method on $m$-split Terwilliger algebra is given in Section~\ref{sec:gener}. Finally, we conclude our work in Section~\ref{sec:conslus}.

\section{Terwilliger algebra and Schrijver's SDP}

Let ${\cal P}$ be the power set of $\{1, \dots, n\}$. A binary code $C$ is a subset of $\cP$. For $X,Y \in \cP$, denote 
{
\begin{equation*}
    X\Delta Y= \{a\in\{1,\dots,n\}: a\in (X \setminus Y) \cup (Y \setminus X) \}.
\end{equation*}}
Let $\lvert S \rvert$ denote the size of a set $S \in \cP$. Hence the (Hamming) distance of $X$ and $Y \in \cP $ is  $\lvert X \Delta Y \rvert$. The distance distribution of the code $C$ is 
\begin{equation}
	A_{j} = \frac{1}{\lvert C \rvert} \sum_{x \in C}\{ y \in C : \lvert x \Delta y \rvert = j\} \label{eq:dist_distriubtion}
\end{equation}
for $j=0,\dots, n$. The minimum distance $d$ of $C$ is the minimum Hamming distance of two distinct elements in $C$ and hence 
$$d= \min\{j>0: A_j>0\}.$$
Note that for $\lvert C \rvert \leq 1$, its minimum distance   is defined to be $\infty$. $C$ is said to be an $(n,d)$ code if $C$ has minimum distance $d$ and length $n$. See more details about codes in~\cite{MS77}.
	
We review the Terwilliger algebra of the Hamming scheme \cite{Ter92,Ter93,BBIT21}. Let $G$ be the group of all distance-preserving automorphisms of ${\cal P}$. Consider the action of $G$ on ${\cal P} \times {\cal P}$ defined by $g(X, Y) = (g X, g Y)$ for $g \in G$  and  $(X, Y) \in {\cal P} \times {\cal P}$ with orbits ${\cal O}_{1}, \dots, {\cal O}_{m}$ for some $m$. For each ${\cal O}_{u}$, we define a $\lvert {\cal P} \rvert \times \lvert {\cal P} \rvert$ matrix $M_{{\cal O}_{u}}$, indexed by the elements in $\cP$, as
\begin{equation*}
	(M_{{\cal O}_{u}})_{X, Y} = \left\{
	\begin{aligned}
		&1, \text{ if } (X, Y) \in {\cal O}_{u};\\
		&0, \text{ otherwise.}
	\end{aligned} \right.
\end{equation*}
Observe that $(X, Y)$ and $(U, V)\in {\cal P} \times {\cal P}$ belong to the same orbit if and only if there is an automorphism $g\in G$ such that $gX = U$ and $gY = V$, that is, if and only if $\lvert X \rvert = \lvert U \rvert,$ $\lvert Y \rvert = \lvert V \rvert,$ \text{ and } $\lvert X \Delta Y \rvert = \lvert U \Delta V \rvert$. Thus, $M_{{\cal O}_{u}}$ can be rewritten as
\begin{equation*}
	(M_{i, j}^{t})_{X, Y} = \left\{
	\begin{aligned}
		&1, \text{ if } \lvert X \rvert = i, \lvert Y \rvert = j, \lvert X \cap Y \rvert = t;\\
		&0, \text{ otherwise,}
	\end{aligned} \right.
\end{equation*}
where each orbit ${\cal O}_{u}$ is represented by some $(i,j,t)$ for $i, j, t \in \{0, \dots, n\}$ with $i+j-2t \in \{0, \dots, n\}$. Let ${\cal A}_{n}$ be the collection of all linear combinations of $\{M_{i, j}^{t}\}$ over the complex field $\mathbb{C}$.  Then ${\cal A}_{n}$ is closed under matrix multiplication and adjoint. Moreover, ${\cal A}_{n}$ is a $\mathbb{C}^{*}$-algebra, called the \textit{Terwilliger algebra} of the Hamming scheme. ${\cal A}_{n}$ is finitely generated with dimension $\binom{n+3}{3}$. 
	
Schrijver described a block diagonal formula for the Terwilliger algebra of the Hamming scheme. 
\begin{theorem} \cite[Theorem 1]{Sch05} \label{thm:ter_blo}
	There is an isomorphism from ${\cal A}_{n}$ to $\bigoplus_{k = 0}^{\lfloor \frac{n}{2} \rfloor} \mathbb{C}^{(n-2k+1) \times (n-2k+1)}$ that maps $A = \sum_{i, j, t}x_{i, j}^{t}M_{i, j}^{t} \in {\cal A}_{n}$ to $\bigoplus_{k = 0}^{\lfloor \frac{n}{2} \rfloor} B_{k}$, where
	\begin{equation*}
		B_{k} = \left( \sum_{t} \binom{n-2k}{i-k}^{-\frac{1}{2}} \binom{n-2k}{j-k}^{-\frac{1}{2}} \beta_{i, j, k}^{t} x_{i, j}^{t} \right)_{i, j = k}^{n-k}
	\end{equation*}
	with
	\begin{equation*}
		\beta_{i, j, k}^{t} = \sum_{u = 0}^{n} (-1)^{u-t} \binom{u}{t} \binom{n-2k}{u-k}\binom{n-k-u}{i-u}\binom{n-k-u}{j-u}.
	\end{equation*}
\end{theorem}
The formula says ${\cal A}_{n}$ as the direct sum of matrices. Therefore, we can represent the elements of ${\cal A}_{n}$ on a computer with a minimal memory.
	
Now, Schrijver's semidefinite constraints for a nontrivial code $C$~\cite{Sch05}  can be derived as follows. Consider the action of $G$. Let $\Pi = \{\pi \in G \mid \emptyset \in \pi(C)\}$ and $\Pi' = \{\pi \in G \mid \emptyset \notin \pi(C)\}$. Let $\chi^{\pi(C)}$ be the incidence vector (as a column vector) of $\pi(C)$ indexed by ${\cal P}$. Define $ \lvert {\cal P} \rvert \times \lvert {\cal P} \rvert$ matrices
\begin{align*}
	&R = \frac{1}{\lvert \Pi \rvert} \sum_{\pi \in \Pi} \chi^{\pi(C)} (\chi^{\pi(C)})^{T}, \\
	&R' = \frac{1}{\lvert \Pi' \rvert} \sum_{\pi \in \Pi'} \chi^{\pi(C)} (\chi^{\pi(C)})^{T}.
\end{align*}
It is obvious that $R$ and $R'$ are positive semidefinite. Let
\begin{equation*}
	x_{i, j}^{t} = \frac{1}{\lvert C \rvert \binom{n}{i-t, j-t, t}} \lambda_{i, j}^{t},
\end{equation*}
where

$\binom{n}{a, b, c} = \frac{n!}{a!b!c!}$ for $a, b, c \geq 0$ with $a+b+c \leq n$, and

\begin{align*}
	\lambda_{i, j}^{t} = &\lvert \{ (X, Y, Z) \in C^{3} : \lvert X \Delta 
	Y \rvert = i, \lvert X \Delta Z \rvert = j, \lvert (X \Delta Y) \cap (X \Delta Z)\rvert = t\}\rvert,
\end{align*}
for each $i, j, t \in \{0, \dots, n\}$ with $i-t \geq 0$, $j-t \geq 0$, and $i+j-2t \leq n$. $\lambda_{i, j}^{t}$ counts the number of triple codewords in $C$ satisfying certain distance relations.  The following proposition says that $R,R'\in\cA_n$.
\begin{proposition} \cite[Proposition $1$]{Sch05}
	\begin{align*}
		& R = \sum_{i, j, t} x_{i, j}^{t} M_{i, j}^{t} \text{ and }\\
		& R' = \frac{\lvert C \rvert}{2^{n}-\lvert C \rvert}\sum_{i, j, t}(x_{i+j-2t, 0}^{0}-x_{i, j}^{t})M_{i, j}^{t}.
	\end{align*}
\end{proposition}
By Theorem \ref{thm:ter_blo}, $R$ and $R'$ are  positive semidefinite if and only if for $k = 0, \dots, \lfloor \frac{n}{2} \rfloor$, the following matrices
\begin{align}
	& \left( \sum_{t} \binom{n-2k}{i-k}^{-\frac{1}{2}} \binom{n-2k}{j-k}^{-\frac{1}{2}} \beta_{i, j, k}^{t} x_{i, j}^{t} \right)_{i, j = k}^{n-k}, \label{sdp_1}\\
	& \left( \sum_{t} \binom{n-2k}{i-k}^{-\frac{1}{2}} \binom{n-2k}{j-k}^{-\frac{1}{2}} \beta_{i, j, k}^{t} (x_{i+j-2t, 0}^{0}-x_{i, j}^{t}) \right)_{i, j = k}^{n-k} \label{sdp_2}
\end{align}
are positive semidefinite. Schrijver also showed that $x_{i,j}^t$ satisfy the following constraints. 
\begin{proposition} \cite{Sch05}
	\label{pro:sch}
	Let $C$ be a code with length $n$ and minimum distance at least $d$. Then $x_{i, j}^{t}$'s corresponding to $C$ satisfy the following constraints:
	\begin{equation}
		\begin{array}{cl}
			\mbox{(i)}     &  x_{0, 0}^{0} = 1;\\
			\text{(ii)}    & 0 \leq x_{i, j}^{t} \leq x_{i, 0}^{0};\\
			\text{(iii)}    & x_{i,0}^{0} + x_{j,0}^{0} \leq  1+x_{i, j}^{t};\\
			\text{(iv)}    & x_{i, j}^{t} = x_{i', j'}^{t'}  \mbox{ if } (i, j, i+j-2t) \mbox{ is a permutation of } (i', j', i'+j'-2t');\\
			\text{(v)}    & x_{i, j}^{t} = 0 \mbox{ if }\{i, j, i+j-2t\} \cap \{1, \dots, d-1\} \neq \emptyset;
		\end{array} \label{sdp_3}
	\end{equation}
	Also, we have
	\begin{equation*}
		\lvert C \rvert = \sum_{i}\binom{n}{i}x_{i, 0}^{0}.
	\end{equation*}
\end{proposition}
Note that $\{\binom{n}{i}x_{i, 0}^{0}\}$  is the distance distribution of $C$. Constraints (i) and (iv) are from the definition directly. Consider  $\lvert X \rvert = i$  and $\lvert Y \rvert = j$ and then we have
\begin{align*}
	&(R)_{X,X} = x_{i, i}^{i} = x_{i, 0}^{0}, \quad (R)_{X,Y} = x_{i, j}^{t}.
\end{align*}
The first inequality of Constraint  (ii) follows  because  of the non-negativity of $\lambda_{i,j}^t$ and the second inequality is because $R$ is positive semidefinite and a diagonal element of a positive semidefinite matrix would dominate its row entries.  Constraint (iii) can be similarly derived from $(R)_{X,X}'\geq (R)_{X,Y}'$ and (iv). To see this, we consider $\lvert X \rvert = i$ and $\lvert Y \rvert =j'= i+j-2t$ at $t'= i-t$. Then
\begin{align*}
	&  (R')_{X, X} = x_{0, 0}^{0} - x_{i, 0}^{0}, \\
	& (R')_{X, Y} = x_{i+j'-2t', 0}^{0} - x_{i, j'}^{t'}= x_{j, 0}^{0} - x_{i, i+j-2t}^{i-t}= x_{j, 0}^{0} - x_{i, j}^{t},
\end{align*}
where $x_{i, i+j-2t}^{i-t}=  x_{i, j}^{t}$ is because of (iv).
Constraint (v) is the requirement  from the minimum distance of the code. 
To sum up, Schrijver's  SDP is as follows with variables $ x_{i, j}^{t} \in\mathbb{R}$:
\begin{align}
	{\rm maximize}&\sum_{i}\binom{n}{i}x_{i, 0}^{0}  \notag\\
	{\rm subject\ to\ }       &  \mbox{ positive semidefiniteness of }(\ref{sdp_1}) \mbox{ and } (\ref{sdp_2})\notag\\
	&(\ref{sdp_3}).    \notag
\end{align}
	
\section{Split Terwilliger algebra of the Hamming scheme} \label{sec:split}
In this section, we consider split distance distribution on a partition of $\{1, \dots, n\}$ with two subsets $T_{1}$ and $T_{2}$ such that $T_{1}\cap T_{2}=\emptyset$, $\lvert T_{1}\rvert = n_{1}$,  $\lvert T_{2}\rvert = n_{2}$, and $n_1+n_2=n$. 
	
Let $G_{u}$ be the group of all distance-preserving automorphisms of the power set of $T_{u}$, for $u = 1, 2$. We consider the group $G_{1} \times G_{2}$ acting on ${\cal P}$ by
\begin{equation*}
	(g, h) \cdot X = g(X \cap T_{1}) \cup h(X \cap T_{2})
\end{equation*}
for  $(g, h) \in G_{1} \times G_{2}$ and  $X \in {\cal P}$. Observe that an orbit of $G_{1} \times G_{2}$ on ${\cal P} \times {\cal P}$ can be similarly represented by $(i, j, t, i', j', t')$ with $i, j, t \in \{0, \dots, n_{1}\}$, $i', j', t' \in \{0, \dots, n_{2}\}$, $i+j-2t \in \{0, \dots, n_{1}\}$ and $i'+j'-2t' \in \{0, \dots, n_{2}\}$ and we define $\lvert {\cal P} \rvert \times \lvert {\cal P}\rvert$ matrices $M_{i, j, i', j'}^{t, t'}$ by
\begin{equation*}
	\left(M_{i, j, i', j'}^{t, t'}\right)_{X,Y}  = \left\{
	\begin{aligned}
		&1, \quad 
		\begin{aligned}
		    &\text{if } \lvert X \cap T_{1}\rvert = i, \lvert X \cap T_{2} \rvert = i',\\
		    &\lvert Y \cap T_{1}\rvert = j, \lvert Y \cap T_{2}\rvert = j', \\
		    &\text{and } \lvert X \cap Y \cap T_{1}\rvert = t, \lvert X \cap Y \cap T_{2}\rvert = t';
	    \end{aligned}\\
	&0, \quad \text{otherwise.}
	\end{aligned} \right.
\end{equation*}
Let ${\cal A}_{n_{1}, n_{2}}$ be the collection of all linear combinations of $\{M_{i, j, i', j'}^{t, t'}\}$ over the complex field $\mathbb{C}$. One can verify that ${\cal A}_{n_{1}, n_{2}}$ is closed under matrix multiplication and adjoint, and ${\cal A}_{n_{1}, n_{2}}$ forms a $\mathbb{C}^{*}$-algebra, which we call a \textit{2-split Terwilliger algebra} of the Hamming scheme. By definition, we have the following lemma.
\begin{lemma} \label{lem:le_1}
	The algebra ${\cal A}_{n_{1}, n_{2}}$ is isomorphic to the algebra ${\cal A}_{n_{1}} \otimes {\cal A}_{n_{2}}$, where $\otimes $ is the matrix tensor product. 
\end{lemma}
	
\begin{proof}
	We denote $M_{i, j}^{n_{1}, t}$'s and $M_{i', j'}^{n_{2}, t'}$'s as the generators of ${\cal A}_{n_{1}}$, ${\cal A}_{n_{2}}$, respectively, for $i, j, t \in \{0, \dots, n_{1}\}$, $i', j', t' \in \{0, \dots, n_{2}\}$. For $X, Y \in {\cal P}$,
	\begin{equation}
		(M_{i, j, i', j'}^{t, t'})_{X, Y} = 1  \label{eq:Mijt_split}
	\end{equation}
	if and only if $\lvert X \cap T_{1}\rvert = i$, $\lvert Y \cap T_{1}\rvert = j$, $\lvert X \cap Y \cap T_{1}\rvert = t$ and $\lvert X \cap T_{2}\rvert = i'$, $\lvert Y \cap T_{2}\rvert = j'$, $\lvert X \cap Y \cap T_{2}\rvert = t'$.\\
	Suppose that the generators of ${\cal A}_{n_{1}}$ and ${\cal A}_{n_{2}}$ are indexed by the power sets ${\cal P}_{1}$, ${\cal P}_{2}$ of $T_{1}$ and $T_{2}$, respectively. We observe that  (\ref{eq:Mijt_split}) is equivalent to
	\begin{equation*}
		(M_{i, j}^{n_{1}, t})_{X\cap T_{1}, Y\cap T_{1}} = 1 \text{ and } (M_{i', j'}^{n_{2}, t'})_{X\cap T_{2}, Y\cap T_{2}} = 1.
	\end{equation*}
	Also, the size of the matrix $M_{i, j}^{n_{1}, t} \otimes M_{i', j'}^{n_{2}, t'}$ is equal to $2^{n_{1}+n_{2}} \times 2^{n_{1}+n_{2}}$. Immediately, we have
	\begin{equation}
		M_{i, j, i', j'}^{t, t'} = M_{i, j}^{n_{1}, t} \otimes M_{i', j'}^{n_{2}, t'}, \label{eq:Mijt_prod}
	\end{equation}
	since the right (left) hand side of (\ref{eq:Mijt_prod}) forms a set of generators for ${\cal A}_{n_{1}, n_{2}}$ (${\cal A}_{n_{1}} \otimes {\cal A}_{n_{2}}$).
	\end{proof}
As a consequence, we have a block-diagonal formula for  the $2$-split Terwilliger algebra of the Hamming scheme, which is an extension of the block-diagonal formula for the Terwilliger algebra of the Hamming scheme with the definition of tensor product. As in Schrijver's decomposition of the Terwilliger algebra, this split decomposition is irreducible. By certain basic results from representation theory,  ${\cal A}_{n_{1}, n_{2}}$ will be mapped to a direct sum of simple ${\cal A}_{n_{1}, n_{2}}$-modules. Thus we have the following corollary.
	
\begin{corollary}\cite[equation $(56)$]{Sch05}
\label{cor:blo}
	There is an isomorphism from ${\cal A}_{n_{1}, n_{2}}$ to
	\begin{equation*}
		\bigoplus_{k = 0}^{\left\lfloor \frac{n_{1}}{2} \right\rfloor} \bigoplus_{k' = 0}^{\left\lfloor \frac{n_{2}}{2} \right\rfloor} \mathbb{C}^{N_{k, k'} \times N_{k. k'}}, 
	\end{equation*}
	with $N_{k, k'} = (n_{1}-2k+1)(n_{2}-2k'+1)$ that maps $A = \sum_{i, j, i', j', t, t'}x_{i, j, i', j'}^{t, t'} M_{i, j, i', j'}^{t, t'}$ to
	\begin{equation*}
		\bigoplus_{k = 0}^{\lfloor \frac{n_{1}}{2} \rfloor} \bigoplus_{k' = 0}^{\lfloor \frac{n_{2}}{2} \rfloor} {\cal B}_{k, k'}, 
	\end{equation*}
	where
	\begin{equation*}
		\begin{aligned}
			{\cal B}_{k, k'} = \left( \sum_{t, t'}  \alpha_{i, j, i', j'}^{k, n_{1}, n_{2}} \beta_{i, j, k}^{n_{1}, t} \beta_{i', j', k'}^{n_{2}, t'}  x_{i, j, i', j'}^{t, t'} \right)_{((i, i'), (j, j')) = ((k, k'), (k, k'))}^{((n_{1}-k, n_{2}-k'), (n_{1}-k, n_{2}-k'))}
		\end{aligned} 
	\end{equation*}
	with
	\begin{align*}
		\beta_{i, j, k}^{n_{l}, t} = & \sum_{u = 0}^{n_{l}} (-1)^{u-t} \binom{u}{t} \binom{n_{l}-2k}{u-k} \binom{n_{l}-k-u}{i-u} \binom{n_{l}-k-u}{j-u},
	\end{align*}
	for $l = 1, 2$ and
	\begin{equation*}
		\alpha_{i, j, i', j'}^{k, n_{1}, n_{2}} = \binom{n_{1}-2k}{i-k}^{-\frac{1}{2}} \binom{n_{1}-2k}{j-k}^{-\frac{1}{2}} \binom{n_{2}-2k'}{i'-k'}^{-\frac{1}{2}} \binom{n_{2}-2k'}{j'-k'}^{-\frac{1}{2}}.
	\end{equation*}
\end{corollary}

\begin{remark}
The formula in Corollary \ref{cor:blo} was provided by Schrijver in \cite{Sch05} to provide SDP constraints for a constant-weight code of weight $w$ by choosing $n_{1}=w$ and $n_{2}=n-w$. 
\end{remark}

Now we can derive additional semidefinite constraints on a nontrivial code $C\subset \cP$ from the 2-split Terwilliger algebra of the Hamming scheme.
Define
\begin{equation}
	x_{i, j, i', j'}^{t, t'} = \frac{1}{\lvert C\rvert \binom{n_{1}}{i-t, j-t, t} \binom{n_{2}}{i'-t', j'-t', t'}} \lambda^{t, t'}_{i, j, i', j'},
\end{equation}
where
\begin{equation}
	\begin{aligned}
		\lambda^{t, t'}_{i, j, i', j'} = \{&(X, Y, Z) \in C^{3} : \lvert (X \Delta Y)\cap T_{1}\rvert = i, \lvert (X \Delta Z)\cap T_{1}\rvert = j, \\
		&\lvert((X \Delta Y)\cap (X \Delta Z))\cap T_{1}\rvert = t,\\
		&\lvert(X \Delta Y)\cap T_{2}\rvert = i', \lvert(X \Delta Z)\cap T_{2}\rvert = j',\\
		&\lvert((X \Delta Y)\cap (X \Delta Z))\cap T_{2}\rvert = t' \}\rvert,
	\end{aligned}
\end{equation}
for $i, j, t \in \{0, \dots, n_{1}\}$, $i', j', t' \in \{0, \dots, n_{2}\}$ with $i-t \geq 0$, $j-t \geq 0$, $i+j-2t \leq n_{1}$, $i'-t' \geq 0$, $j'-t' \geq 0$ and $i'+j'-2t' \leq n_{2}$. 
Also,  the size of the code is 
\begin{equation*}
	\lvert C \rvert=\sum_{a = 0}^{n} \sum_{i+i' = a} \binom{n_{1}}{i}\binom{n_{2}}{i'} x_{i, 0, i', 0}^{0, 0}.
\end{equation*}
	
Let $\lambda_{a, b}^{c}$, $x_{a, b}^{c}$ and $M_{a, b}^{c}$ be defined as in Schrijver's SDP in the previous section. From the definitions, we have the following identities:
\begin{align}
	\lambda_{a, b}^{c} =& \sum_{i+i' = a, j+j' = b, t+t' = c} \lambda_{i, j, i', j'}^{t, t'},   \label{eqn:id_1}\\
	x_{a, b}^{c} =& \sum_{i+i' = a, j+j' = b, t+t' =c} \frac{\binom{n_{1}}{i-t, j-t, t} \binom{n_{2}}{i'-t', j'-t', t'}}{\binom{n}{a-c, b-c, c}}x_{i, j, i', j'}^{t, t'},  \label{eqn:id_2}\\
	M_{a, b}^{c} =& \sum_{i+i' = a, j+j' = b, t+t' = c} M_{i, j, i', j'}^{t, t'}.  \label{eqn:id_3}
\end{align}
	
Consider the following two sets of automorphisms: 
\begin{align*}
	& \Pi_{\rm s} = \{(\pi_{1}, \pi_{2}) \in G_{1} \times G_{2} \mid \emptyset \in (\pi_{1}, \pi_{2})(C)\},\\
	& \Pi_{\rm s}' = \{(\pi_{1}, \pi_{2}) \in G_{1} \times G_{2} \mid \emptyset \notin (\pi_{1}, \pi_{2})(C)\}.
\end{align*}
Here the subscript $\rm s$ stands for split. Similarly, we define
\begin{align*}
	& R_{\rm s} = \frac{1}{\lvert \Pi_{\rm s} \rvert} \sum_{(\pi_{1}, \pi_{2}) \in \Pi_{\rm s}} \chi^{(\pi_{1}, \pi_{2})(C)} (\chi^{(\pi_{1}, \pi_{2})(C)})^{T}, \\
	& R_{\rm s}' = \frac{1}{\lvert \Pi_{\rm s}'\rvert}\sum_{(\pi_{1}, \pi_{2}) \in \Pi_{\rm s}'} \chi^{(\pi_{1}, \pi_{2})(C)} (\chi^{(\pi_{1}, \pi_{2})(C)})^{T}.
\end{align*}
One can immediately see that $R_{\rm s}$ and $R_{\rm s}'$ are positive semidefinite and they only depend on the action $G_{1} \times G_{2}$ on $C$. Moreover, $R_{\rm s}$ and $R_{\rm s}'$ are elements of ${\cal A}_{n_{1}, n_{2}}$ as a consequence of the following proposition.
\begin{proposition}
	\label{pro:R_R'}
	\begin{equation*}
		\begin{aligned}
			&R_{\rm s} = \sum_{i, j, i', j', t ,t'} x_{i, j, i', j'}^{t, t'}M_{i, j, i', j'}^{t, t'},\\
			& 
			\begin{aligned}
				R_{\rm s}' =& \frac{\lvert C \rvert}{2^{n}-\lvert C\rvert}\sum_{i, j, i', j', t ,t'} \left( x_{i+j-2t, 0, i'+j'-2t', 0}^{0, 0}-x_{i, j, i', j'}^{t, t'}\right) \\
				&\cdot M_{i, j, i', j'}^{t, t'}.
			\end{aligned}
		\end{aligned}
	\end{equation*}
\end{proposition}
	
\begin{proof}
	Let $X \in {\cal P}$. We define $\Pi_{s}^{X} = \{(\pi_{1}, \pi_{2}) \in \Pi_{1} \times \Pi_{2} \mid (\pi_{1}, \pi_{2})(X) = \emptyset\}$. For an element $(\pi_{1}, \pi_{2}) \in \Pi_{s}^{X}$, we can see $\pi_{1}$ as a permutation of $T_{1}$ and $\pi_{2}$ as a permutation of $T_{2}$. Hence, we have $\lvert \Pi_{s}^{X}\rvert = n_{1} ! n_{2} !$. Define
	\begin{align*}
		& R_{s}^{X} = \frac{1}{\lvert \Pi_{s}^{X} \rvert} \sum_{(\pi_{1}, \pi_{2}) \in \Pi_{s}^{X}} \chi^{(\pi_{1}, \pi_{2})(C)} (\chi^{(\pi_{1}, \pi_{2})(C)})^{T}, \\
		&
		\begin{aligned}
			\lambda^{t, t', X}_{i, j, i', j'} = \lvert\{&(Y, Z) \in C^{2} : \lvert(X \Delta Y)\cap T_{1}\rvert = i, \lvert(X \Delta Z)\cap T_{1}\rvert = j,\\ 
			&\lvert((X \Delta Y)\cap (X \Delta Z))\cap T_{1}\rvert = t,\\
			&\lvert(X \Delta Y)\cap T_{2}\rvert = i', \lvert(X \Delta Z)\cap T_{2}\rvert = j',\\
			&\lvert((X \Delta Y)\cap (X \Delta Z))\cap T_{2}\rvert = t' \}\rvert.
		\end{aligned}
	\end{align*}
	For $(\pi_{1}, \pi_{2}) \in \Pi_{s}^{X}$ and fixed $i, j, t, i', j', t'$, observe that the number of $1$'s in 
	\begin{equation*}
		(\chi^{(\pi_{1}, \pi_{2})(C)} (\chi^{(\pi_{1}, \pi_{2})(C)})^{T})_{Y, Z}
	\end{equation*} 
	such that $(M_{i, j, i', j'}^{t, t'})_{Y, Z} = 1$ is $\lambda^{t, t', X}_{i, j, i', j'}$. There are $\binom{n_{1}}{i-t, j-t, t} \binom{n_{2}}{i'-t', j'-t', t'}$ such $(Y, Z)$. Thus,
	\begin{equation*}
		R_{s}^{X} = \sum_{i, j, i', j', t, t'} \frac{1}{\binom{n_{1}}{i-t, j-t, t} \binom{n_{2}}{i'-t', j'-t', t'}} \lambda^{t, t', X}_{i, j, i', j'} M_{i, j, i', j'}^{t, t'}.
	\end{equation*}
	Next, we see that
	\begin{equation*}
		R_{s} = \sum_{X \in C} \frac{R_{s}^{X}}{\lvert C\rvert}, \quad R_{s}' = \sum_{X \notin C} \frac{R_{s}^{X}}{\lvert {\cal P} \setminus C \rvert},
	\end{equation*}
	and
	\begin{equation*}
		\sum_{X \in C} \lambda^{t, t', X}_{i, j, i', j'} = \lambda^{t, t'}_{i, j, i', j'}.
	\end{equation*}
	Therefore, 
	\begin{equation*}
		\begin{aligned}
			R_{s} =& \sum_{X \in C} \frac{R_{s}^{X}}{\lvert C \rvert}\\
			=& \frac{1}{\lvert C \rvert} \sum_{i, j, i', j', t, t'} \frac{1}{\binom{n_{1}}{i-t, j-t, t} \binom{n_{2}}{i'-t', j'-t', t'}} \left(\sum_{X \in C} \lambda^{t, t', X}_{i, j, i', j'}\right) M_{i, j, i', j'}^{t, t'}\\
			=& \sum_{i, j, i', j', t, t'} \frac{\lambda^{t, t'}_{i, j, i', j'}}{\lvert C\rvert \binom{n_{1}}{i-t, j-t, t} \binom{n_{2}}{i'-t', j'-t', t'}} M_{i, j, i', j'}^{t, t'}\\
			=& \sum_{i, j, i', j', t, t'} x_{i, j, i', j'}^{t, t'} M_{i, j, i', j'}^{t, t'}.
		\end{aligned}
	\end{equation*}
	For $(Y, Z) \in C^{2}$ with $\lvert Y \Delta Z \cap T_{1}\rvert = i+j-2t$, $\lvert Y \Delta Z \cap T_{2}\rvert = i'+j'-2t'$. The number of $X \in {\cal P}$ such that $\lvert (X \Delta Y)\cap T_{1}\rvert = i, \lvert(X \Delta Z)\cap T_{1}\rvert = j, \lvert((X \Delta Y)\cap (X \Delta Z))\cap T_{1}\rvert = t, \lvert(X \Delta Y)\cap T_{2}\rvert = i', \lvert(X \Delta Z)\cap T_{2}\rvert = j', \lvert((X \Delta Y)\cap (X \Delta Z))\cap T_{2}\rvert = t'$ is $\binom{i+j-2t}{i-t}\binom{n_{1}-t-j+2t}{t}\binom{i'+j'-2t'}{i'-t'}\binom{n_{2}-i'-j'+2t'}{t'}$.\\
	Thus we have
	\begin{equation*}
		\begin{aligned}
			\sum_{X \in {\cal P}} \lambda_{i, j, i', j'}^{t, t', X} = &\binom{i+j-2t}{i-t}\binom{n_{1}-t-j+2t}{t}\binom{i'+j'-2t'}{i'-t'} \binom{n_{2}-i'-j'+2t'}{t'} \lambda_{i+j-2t, 0, i'+j'-2t', 0}^{0, 0}.
		\end{aligned}
	\end{equation*}
	Hence, 
	\begin{equation*}
		\begin{aligned}
			R_{s}' =& \sum_{X \in {\cal P} \setminus C} \frac{R_{s}^{X}}{2^{n}-\lvert C \rvert}\\
			=& \frac{1}{2^{n}-\lvert C \rvert} \sum_{i, j, i', j', t, t'} \frac{1}{\binom{n_{1}}{i-t, j-t, t}\binom{n{2}}{i'-t', j'-t', t'}} \left( \sum_{X \in {\cal P} \setminus C} \lambda_{i, j, i', j'}^{t, t', X} \right) M_{i, j, i', j'}^{t, t'}\\
			=& \frac{\lvert C \rvert}{2^{n}-\lvert C \rvert} \sum_{i, j, i', j', t, t'} \frac{1}{\lvert C \rvert \binom{n_{1}}{i-t, j-t, t}\binom{n{2}}{i'-t', j'-t', t'}}\\
			& \cdot \left( \binom{i+j-2t}{i-t}\binom{n_{1}-t-j+2t}{t} \right. \binom{i'+j'-2t'}{i'-t'}\\
			& \cdot \binom{n_{2}-i'-j'+2t'}{t'} \left. \lambda_{i+j-2t, 0, i'+j'-2t', 0}^{0, 0} - \lambda_{i, j, i', j'}^{t, t'} \right) M_{i, j, i', j'}^{t, t'}\\
			=& \frac{\lvert C \rvert}{2^{n}-\lvert C\rvert} \sum_{i, j, i', j', t, t'} \left( \frac{\binom{i+j-2t}{i-t}\binom{n_{1}-t-j+2t}{t} \binom{i'+j'-2t'}{i'-t'}\binom{n_{2}-i'-j'+2t'}{t'}}{\lvert C \rvert \binom{n_{1}}{i-t, j-t, t}\binom{n_{2}}{i'-t', j'-t', t'}} \right. \\
			&\cdot \lambda_{i+j-2t, 0, i'+j'-2t', 0}^{0, 0}  - \left. x_{i, j, i', j'}^{t, t'} \right) M_{i,j, i', j'}^{t, t'}.
		\end{aligned}
	\end{equation*}
	Using the identities
	\begin{align*}
		&\binom{n_{1}}{i-t, j-t, t}^{-1} \binom{i+j-2t}{i-t}\binom{n_{1}-i-j+2t}{t} = \binom{n_{1}}{i+j-2t}^{-1} \text{ and }\\
		&\binom{n_{2}}{i'-t', j'-t', t'}^{-1} \binom{i'+j'-2t'}{i'-t'}\binom{n_{2}-i'-j'+2t'}{t'} = \binom{n_{2}}{i'+j'-2t'}^{-1},
	\end{align*}
	we have
	\begin{equation*}
		\begin{aligned}
			R_{s}' =& \frac{\lvert C\rvert}{2^{n}-\lvert C\rvert} \sum_{i, j, i', j', t, t'} \left(\frac{1}{\lvert C\rvert \binom{n_{1}}{i+j-2t}\binom{n_{2}}{i'+j'-2t'}} \right. \left.\lambda_{i+j-2t, 0, i'+j'-2t', 0}^{0, 0} - x_{i, j, i', j'}^{t, t'} \right) M_{i, j, i', j'}^{t, t'}\\
			=& \frac{\lvert C \rvert}{2^{n}-\lvert C\rvert} \sum_{i, j, i', j', t, t'} \left( x_{i+j-2t, 0, i'+j'-2t', 0}^{0, 0}-x_{i, j, i', j'}^{t, t'}\right) M_{i, j, i', j'}^{t, t'}.
		\end{aligned}
	\end{equation*}
\end{proof}

Using Theorem~\ref{cor:blo} and the positive semidefiniteness of $R_{\rm s}$ and $R_{\rm s}'$, we have the following semidefinite constraints. 
For each $k = 0, \dots, \lfloor \frac{n_{1}}{2} \rfloor$ and $k' = 0, \dots, \lfloor \frac{n_{2}}{2} \rfloor$, the following matrices are positive semidefinite: 
\begin{align}
	\left( \sum_{t, t'} \alpha_{i, j, i', j'}^{k, n_{1}, n_{2}} \beta_{i, j, k}^{n_{1}, t} \beta_{i', j', k'}^{n_{2}, t'} x_{i, j, i', j'}^{t, t'}  \right)_{((i, i'), (j, j')) = ((k, k'), (k, k'))}^{((n_{1}-k, n_{2}-k'), (n_{1}-k, n_{2}-k'))}, \label{sem1}
\end{align} 
\begin{align}
	\left( \sum_{t, t'} \alpha_{i, j, i', j'}^{k, n_{1}, n_{2}}  \beta_{i, j, k}^{n_{1}, t} \beta_{i', j', k'}^{n_{2}, t'} (x_{i+j-2t, 0, i'+j'-2t', 0}^{0, 0} -x_{i, j, i', j'}^{t, t'}) \right)_{((i, i'), (j, j')) = ((k, k'), (k, k'))}^{((n_{1}-k, n_{2}-k'), (n_{1}-k, n_{2}-k'))}. \label{sem2}
\end{align}

\begin{proposition}
	\label{pro:lin}
	Let $C$ be a code with length $n$ and minimum distance at least $d$. Then 
	$x_{i,j,i',j'}^{t,t'}$	corresponding to $C$ satisfy the following linear constraints: for $i, j, t, a, b, c \in \{0, \dots,$ $n_{1}\}$ and $i', j', t', a', b', c' \in \{0, \dots, n_{2}\}$,
	\begin{align}
	    \begin{array}{cl}
		    \mbox{(i) }& x_{0, 0, 0, 0}^{0, 0} = 1; \\
			\mbox{(ii) }& 0 \leq x_{i, j, i', j'}^{t, t'} \leq x_{i, 0, i', 0}^{0, 0};  \\
			\mbox{(iii) }& x_{i, 0, i', 0}^{0, 0} + x_{0, j, 0, j'}^{0, 0} \leq 1 + x_{i, j, i', j'}^{t, t'}; \\
			\mbox{(iv) }& x_{i, j, i', j'}^{t, t'} = x_{a, b, a', b'}^{c, c'} \mbox{ if } ((i, i'), (j, j'), (i+j-2t, i'+j'-2t'));  \\
			&\mbox{ is a permutation of } ((a, a'), (b, b'), (a+b-2c, a'+b'-2c')),  \\
			\mbox{(v) }&  x_{i, j, i', j'}^{t, t'} = 0 \mbox{ if } \{i+i', j+j', (i+i')+(j+j')-2(t+t')\} \cap \{1, \dots, d-1\} \neq \emptyset. \label{eq:split_cons}
		\end{array}
	\end{align}
\end{proposition}

Constraints (i) and (iv) come from the definition of $x_{i, j, i', j'}^{t, t'}$. Constraints (ii) and (iii) are because of the positive semidefiniteness of $R$ and $R'$. Constraint (iv) is the requirement of the minimum distance.

The split distance distribution of $C$ with respect to the partition $\{T_1,T_2\}$ is  $\{A_{i, j}\}$, where
\begin{equation*}
	A_{i, j} = \frac{1}{\lvert C \rvert} \lvert \{(a, b) \in C \times C \mid \lvert a \Delta b \cap T_{1}\rvert = i, \lvert a \Delta b \cap T_{2}\rvert = j\}\rvert.
\end{equation*}
It has been shown that generalized Delsarte's inequalities on $A_{i,j}$ hold~\cite{Sim95}.
\begin{corollary}\cite[Generalized Delsarte's inequalities]{Sim95}
	\label{coro:gen_del}
	\begin{equation}
		\sum_{i, j} A_{i, j} K_{p}^{n_{1}}(i) K_{q}^{n_{2}}(j) \geq 0, \label{eq:split_del}
	\end{equation}
	for $p = 0, \dots, n_{1}$ and $q = 0, \dots, n_{2}$, where
	\begin{equation*}
		K_{k}^{n}(x) = \sum_{y=0}^k (-1)^{y} \binom{x}{y} \binom{n-x}{k-y}
	\end{equation*}
	is the binary Krawchuk polynomial.
\end{corollary}
	
We can show that the generalized Delsarte's inequalities on $\{A_{i,j}\}$ can be derived from the positive semidefiniteness of $R_s$ and $R_s'$ in Proposition \ref{pro:R_R'}.

\begin{lemma}
If $R_s$ and $R_s'$ are positive semidefinite, then (\ref{eq:split_del}) holds.
\end{lemma}
	
\begin{proof}
	Define a matrix	
	\begin{equation*}
		D_{s} =  R_{s} + \frac{2^{n}-\lvert C \rvert}{\lvert C \rvert}R_{s}'= \sum_{i, j, i', j', t, t'} x_{i+j-2t, 0, i'+j'-2t', 0}^{0, 0} M_{i, j, i', j'}^{t, t'},
	\end{equation*}
	which   is a nonnegative linear combination of positive semidefinite matrices $R_{s}$ and $R_{s}'$, and hence is positive semidefinite. 
	We can rewrite $D_{s}$ as
	\begin{equation*}
		D_{s} = \sum_{k, k'} x_{k, 0, k', 0}^{0, 0} M_{k, k'},
	\end{equation*}
	where
	\begin{equation*}
		(M_{k, k'})_{X, Y} = \left\{
		\begin{aligned}
			&1, \text{ if } \lvert X \Delta Y \cap T_{1}\rvert = k, \lvert X \Delta Y \cap T_{2}\rvert = k';\\
			&0, \text{ otherwise.}
		\end{aligned} \right.
	\end{equation*}
	Define matrices $D_{a}^{m}$ for $0\leq a,m\leq n$ with entries
	\begin{equation*}
		(D_{a}^{m})_{X, Y} = \left\{
		\begin{aligned}
			&1, \text{ if } \lvert X \Delta Y\rvert = a;\\
			&0, \text{ otherwise,}
		\end{aligned}
		\right.
	\end{equation*}
	for indexes $X, Y$ in the power set of $\{1, \dots, m\}$. Notice that $\{D_{a}^{m}\}$ forms a basis for the Bose–Mesner algebra of the Hamming scheme with length $m$. Using the isomorphism in Lemma \ref{lem:le_1}, one finds that $M_{k, k'}$ is isomorphic to  $D_{k}^{n_{1}} \otimes D_{k'}^{n_{2}}$. By \cite{Del73}, we learn that Krawchuk polynomials $K_{a}^{m}(p)$ are eigenvalues of the matrices $\{D_{a}^{m}\}$ for $0 \leq p \leq m$ and, moreover, the matrices $\{M_{k, k'}\}$ are commutative. Consequently, to diagonalize $D_{s}$,  we  simply have to diagonalize $M_{k, k'}$ and the semidefinite conditions become
	\begin{equation}
		\label{ineq:cor_2}
		\sum_{k, k'} x_{k, 0, k', 0}^{0, 0} K_{k}^{n_{1}}(p) K_{k'}^{n_{2}}(q) \geq 0,
	\end{equation}
	for $p \in \{0, \dots, n_{1}\}$ and $q \in \{0, \dots, n_{2}\}$. Recall that the Krawchuk polynomial obeys the following symmetric relation~\cite{MS77}:
	\begin{equation*}
		\binom{m}{a}K_{b}^{m}(a) = \binom{m}{b}K_{a}^{m}(b).
	\end{equation*}
	Thus the inequality (\ref{ineq:cor_2}) becomes
	\begin{equation*}
		\sum_{k, k'} \binom{n_{1}}{k}\binom{n_{2}}{k'} x_{k, 0, k', 0}^{0, 0} \binom{n_{1}}{p} \binom{n_{2}}{q}K_{k}^{n_{1}}(p) K_{k'}^{n_{2}}(q) \geq 0,
	\end{equation*}
	which implies the result.
\end{proof}

Let ${\cal N}_{k}$ be the collection of codes $S \subset {\cal P}$ with minimum distance at least $d$ and $\lvert S \rvert \leq k$. Schrijver's SDP can be generalized by considering the $k$-points relation of codewords. The generalized version is as follows \cite{GMS12}.

\begin{proposition}\cite[Generalized  SDP bound]{GMS12}
    \label{pro:gen_sch_sdp}
    For $S \in {\cal N}_{k}$, define
    \begin{equation*}
        {\cal N}(S) = \{S' \subset {\cal P} : S \subset S', \lvert S \rvert + 2\lvert S' \setminus S \rvert \leq k \}.
    \end{equation*}
    Let $x = \{x_{i}\}_{i \in {\cal N}_{k}} \subset \mathbb{R}$ be a set of non-negative numbers. Define the $\lvert{\cal N}(S)\rvert \times \lvert{\cal N}(S)\rvert$ matrix indexed by ${\cal N}(S)$ as
    \begin{equation*}
    \label{gen_Sch}
        \left( M_{S}(x)\right)_{S', S''} = \left\{
        \begin{aligned}
            &x_{S' \cup S''}, \quad \text{ if } S' \cup S'' \text{ has minimum distance at least } d;\\
            &0,  \quad \quad \quad \quad \text{otherwise,}
        \end{aligned}\right.
    \end{equation*}
    for $S', S'' \in {\cal N}(S)$. Then,
    \begin{equation}
    \label{equ:k_sdp}
        A(n, d) \leq A_{k}(n, d) = {\rm max}\left\{\sum_{i \in {\cal P}}x_{\{i\}} : x_{\emptyset} = 1, x_{S} \geq 0 \text{ and } M_{S}(x) \text{ is positive semidefinite for each } S \in {\cal N}_{k} \right\}.
    \end{equation}
\end{proposition}

It can be shown that the generalized Schrijver's SDP constraints with $k = 4$ induces our positive semidefinite constraint on $R_{s}$.

\begin{proposition}
       The positive semidefinite constraints in the quadruple semidefinite program in (\ref{equ:k_sdp}) imply the positive semidefiniteness of $R_{s}$.
\end{proposition}

\begin{proof}
    Let $C$ be a code of length $n$ and minimum distance $d$. Consider $S = \{\emptyset, u\}  \in {\cal N}_{4}$,
   where $u$ has $1$ in the first $n_{1}$ positions and $0$ in the other $n_{2}$ positions,  which is a subset of $ C$.  Let $x$ be defined as $x_{i} = 1$ if $i \subset C$ and $x_{i} = 0$ for the other cases. Observe that there is an one-to-one correspondence between $N(S) = \{S \cup \{v\} : v \in {\cal P}\}$ and ${\cal P}$ by sending $S \cup \{v\}$ to $v$. Thus,  $M_{S}(x)$ can be indexed by ${\cal P}$ and  
    \begin{equation*}
        \left( M_{S}(x)\right)_{v, w} = \left\{
        \begin{aligned}
            &1, \quad \text{if } \{v, w\} \subset C;\\
            &0,  \quad \text{otherwise.}
        \end{aligned}\right.
    \end{equation*}
    Notice that $M_{S}(x) = \chi^{C} (\chi^{C})^{T}$ is positive semidefinite and
    \begin{equation*}
        R_{s}^{u} = \frac{1}{\lvert \Pi_{s}^{u}\rvert} \sum_{(\pi_{1}, \pi_{2}) \in \Pi_{s}^{u}} \chi^{(\pi_{1}, \pi_{2})(C)} (\chi^{(\pi_{1}, \pi_{2})(C)})^{T}.
    \end{equation*}
    We observe that each $\chi^{(\pi_{1}, \pi_{2})(C)} (\chi^{(\pi_{1}, \pi_{2})(C)})^{T}$ is similar to $M_{S}(x) = \chi^{C} (\chi^{C})^{T}$ by changing the order of the basis. Thus, $R_{s}^{u}$ is semidefinite.
\end{proof}

\section{Semidefinite program} \label{sec:SDP}
In this section, we provide an SDP  that gives an upper bound on the size of an $(n,d)$ code $C\subset \cP$. We include the  known linear constraints in the literature, which are critical in the SDP. We observe that an SDP with semidefinite constraints based on quadruple distances \cite{GMS12} does not improve the upper bounds of $A(18, 4)$ and $A(19, 4)$. In fact, the bounds obtained with only Schrijver's semidefinite constraints are even worse than the bounds with linear constraints \cite{Del73, Bes80} in some cases. For instance, $A(18, 4) \leq 6552$ can be obtained by certain linear constraints, while the SDP gives us an upper bound of $6553$.
	
\subsection{linear constraints}
The distance distribution $A_j$ of  $C$ in (\ref{eq:dist_distriubtion}) can be represented in terms of $x_{i,i',j,j'}^{t,t'}$ as
\begin{equation*}
	A_{j} = \sum_{i+i' = j} \binom{n_{1}}{i} \binom{n_{2}}{i'} x_{i, 0, i', 0}^{0, 0}. 
\end{equation*}
Furthermore, the propagation rule  $A(n-1, 2e-1) = A(n, 2e)$ \cite{MS77} implies that we  need to consider only even distance $j$.

\begin{lemma}\cite[Lemma $7$]{Bes80}
	\begin{equation}
		\sum_{i = 0}^{n} \binom{n-i}{k}A_{i} \leq \binom{n}{k} A(n-k, d). \label{le3}
	\end{equation} 
\end{lemma}
	
\begin{lemma}\cite[Theorem $9$]{MEL02}
	\begin{equation}
		\label{le4}
		A_{n-\frac{d}{2}}+\left\lfloor \frac{2n}{d} \right\rfloor \sum_{i < \lfloor \frac{d}{2} \rfloor}A_{n-i} \leq \left\lfloor \frac{2n}{d} \right\rfloor.
	\end{equation}
\end{lemma}
	
Our SDP also benefits from bounds on $A(n,d,w)$, the maximum size of a length-$n$ code with minimum distance $d$ and constant weight $w$.
\begin{lemma}\cite[Lemma $5$]{Bes80}
	Let $P = A(n-1, d, \frac{1}{2}d+1)$, $Q = A(n-\frac{1}{2}d, d, \frac{1}{2}d+1)$, $R = A(n-\frac{1}{2}d+2, d, \frac{1}{2}d+2)$, then
	\begin{equation}
		\label{le5}
		\begin{aligned}
			&(\frac{1}{2}d+2)A_{n-\frac{1}{2}d-2}+\frac{1}{2}d(P-Q)A_{n-\frac{1}{2}d} +(n P-(\frac{1}{2}d+2)R)A_{n-\frac{1}{2}d+2}+n P\sum_{i = n-\frac{1}{2}d+3}^{n} A_{i} \leq n P.
		\end{aligned}
	\end{equation}
\end{lemma}
	
\begin{lemma}\cite[Theorem $10$]{MEL02}
	For $i = 1, \dots, \frac{d}{2}-1$, 
	\begin{equation}
		\label{le6}
		\begin{aligned}
			A_{n-\frac{d}{2}-i}+&(A(n, d, \frac{d}{2}+i)-A(n-\frac{d}{2}+i, d, \frac{d}{2}+i))A_{n-\frac{d}{2}+i}+A(n, d, \frac{d}{2}+i)\sum_{j>i}A_{n-\frac{d}{2}+j} \leq A(n, d, \frac{d}{2}+i).
		\end{aligned}
	\end{equation}
\end{lemma}
	
\begin{lemma}\cite[ (25)]{Sch05}
	For $i = 0, \dots, n$,
	\begin{equation}
		\label{le7}
		A_{i} \leq A(n, d, i).
	\end{equation}
\end{lemma}
	
We also have additional linear constraints from doubly constant-weight codes. Let $T(w_{1}, t_{1}, w_{2}, t_{2}, d)$ be the maximum possible size of a doubly constant-weight code, which is a $(t_{1}+t_{2}, d, w_{1}+w_{2})$ constant-weight code such that every codeword has exactly $w_{1}$ and $w_{2}$ ones on the first $t_{1}$ and next $t_{2}$  coordinates, respectively.
	
\begin{lemma}
\cite[Theorem $3$]{KT13}
	For $i, j, t \in \{0, \dots, n\}$, we have
	\begin{equation}
		\label{le14}
		x_{i, j}^{t} \leq \frac{T(t, i, j-t, n-i, d)}{\binom{i}{t} \binom{n-i}{j-t}}x_{i, 0}^{0}.
	\end{equation}
\end{lemma}
	
\subsection{Semidefinite program for binary codes}
Collecting all the mentioned linear and semidefinite constraints, we have the following SDP on $A(n,d)$ with variables $ x_{i, j,i',j'}^{t,t'} \in\mathbb{C}$:
\begin{align}
	{\rm maximize\ } &\sum_{a = 0}^{n} \sum_{i+i' = a} \binom{n_{1}}{i}\binom{n_{2}}{i'} x_{i, 0, i', 0}^{0, 0}.  \notag\\
	{\rm subject\ to\ }       &  \mbox{ positive semidefiniteness of }(\ref{sdp_1}), (\ref{sdp_2}),(\ref{sem1}), (\ref{sem2}) \label{eq:main_SDP} 
	\\ & (\ref{sdp_3}),(\ref{eqn:id_1}), (\ref{eqn:id_2}), (\ref{eq:split_cons})  \notag \\
	& (\ref{le3}), (\ref{le4}), (\ref{le5}), (\ref{le6}), (\ref{le7}), (\ref{le14}).    \notag 
\end{align}
	
As mentioned in the previous section, the generalized Delsarte's inequalities~(\ref{eq:split_del}) are implicitly included in the SDP.

\subsection{The correctness of computer computational results}
\label{subsec:err}
Since numerical methods will be used to approximate the optimal solution to an SDP, when we have a large number of variables, the accuracy of computer simulations may not be sufficient. In \cite{Gij05}, Gijswijt used the weak duality of the optimization program to verify the SDP bounds. We describe his method here.
	
Consider an SDP of the following form:
\begin{align*}
	&\text{maximize} \quad \sum_{i=1}^{m}x_{i}c_{i}\\
	&\text{subject to} \quad \sum_{i=1}^{m}x_{i}F_{i} + F_{0}  \text{ is positive semidefinite}
\end{align*}
with variables $x_{1}, \dots, x_{m}\in\mathbb{R}$, constants $c_{1}, \dots, c_{m}\in\mathbb{R}$ and symmetric matrices $F_{0}, \dots, F_{m}\in\mathbb{R}^{n\times n}$. Its dual problem is as follows: 

\begin{align*}
	\text{minimize} \quad & \text{tr}{F_{0} Y} \\
	\text{subject to} \quad &\text{tr}{F_{i}Y} + c_{i}=0  \text{ for } i=1,\dots, m\\
	& Y \text{ is positive semidefinite.}
\end{align*}
Every feasible $Y$ in the dual problem gives an upper bound of the primal problem. 
	
In a numerical computation, a dual solution $Y$ may not exactly satisfy the constraints, but 
\begin{equation*}
	\text{tr} F_{i} Y + c_{i} = \epsilon_{i}
\end{equation*}
for some small numbers $\epsilon_{1}, \dots, \epsilon_{m}$ due to the computer accuracy. Similarly, a primal solution may not be reliable.
However, we can estimate the computation error as follows. Let $x_{1}^\star, \dots, x_{m}^\star$ be an optimal solution for the primal problem with objective value $P = \sum_{i = 1}^{m}x_{i}^\star c_{i}$. Let $X=\sum_{i=1}^{m}x_{i}^\star F_{i} + F_{0}$, which is positive semidefinite since $x_i^\star$ are feasible. Then
\begin{align*}
	\text{tr} F_{0} Y  &= \text{tr}  \left\{\left(-\sum_{i = 1}^{m} x_{i}^\star F_{i} +X\right) Y \right\}\\
	& \geq \text{tr}  \left\{-\sum_{i = 1}^{m} x_{i}^\star F_{i} Y \right\}\\
	&= -\sum_{i=1}^{m} x_{i}^\star \text{tr}  F_{i} Y  \\
	&= -\sum_{i = 1}^{m} x_{i}^\star(-c_{i}+ \epsilon_{i})\\
	&= P - \sum_{i = 1}^{m} x_{i}^\star\epsilon_{i}.
\end{align*}

In our SDP~(\ref{eq:main_SDP}), the variables $x_{i,j,i',j'}^{t,t'}$ are negative and no larger than one by~(\ref{eq:split_cons}).
For our purpose, we can have an upper bound on the optimal value $P$ that
\begin{equation*}
	P \leq	\text{tr} F_{0} Y + \sum_{i = 1}^{m}x_{i}^{*}\epsilon_{i} \leq  \text{tr} F_{0} Y + \sum_{i=1}^{m}{\rm max}\{0,\epsilon_{i}\}.
\end{equation*}
Therefore, we may use the dual optimal value and error terms to estimate an upper bound in our computational results. 

\subsection{computational results}
We use the CVX toolbox \cite{GB14}, \cite{GB08} in MATLAB with the MOSEK solver  to run our SDP. Our main results are as follows. In our SDP,  we have tested all possible values of splits $n_{1}$ and $n_{2}$. The computational results give us two improvements	$A(18,4) \leq 6551$ and 	$A(19,4) \leq 13087$. However, using Gijswijt's method in the previous subsection, we are only able to ensure one of them is improved.
\begin{theorem}
	$A(18,4) \leq 6551$.
\end{theorem}
\begin{proof}
    Using $n_{1} = 2$, we obtain $A(18,4)\leq 6551.93$ with an error term less than $10^{-16}$. Thus $A(18,4)\leq 6551$.
\end{proof}

\begin{remark}
$A(18,4) \leq 6551$ can be obtained by the split SDP with $n_{1} = 2$ and only the constraints (\ref{sem1}), (\ref{sem2}), (\ref{le3}) and (\ref{le5}). Then, we have $A(18,4) \leq 6551.98$ with an error term less than $10^{-16}$. All the above mentioned constraints are necessary in this case.
\end{remark}

\section{Generalization}
\label{sec:gener}
In this section, we  consider an $m$-split distance distribution defined on  a partition of $\{1, \dots, n\}$ with  arbitrary $m$ subsets, say $T_{1}, \dots, T_{m}$,  each of size $\lvert T_{p}\rvert = n_{p}$ for $p \in \{1, \dots, m\}$. Our method can be generalized in this case to introduce more semidefinite constraints. Proofs to these generalizations are similar to the $2$-split case and will be omitted. Finally, we discuss the underlying association scheme structure.\\

\subsection{Semidefinite constraints from  \texorpdfstring{$m$}--split Terwilliger algebras}
Consider a group $G = \prod_{p = 1}^{m}G_{p}$, where $G_{p}$ is the isometry group on the power set of $T_{p}$. Let $\sigma_{1}, \dots, \sigma_{q}$ be the orbits of $G$ acting on ${\cal P}$ for some $q$. Then we define  $\lvert {\cal P}\rvert \times \lvert {\cal P}\rvert$ matrices, indexed by ${\cal P}$, 
\begin{equation*}
	(M_{\sigma_{k}})_{X, Y} = \left\{
	\begin{aligned}
		&1, \text{ if } X, Y \in \sigma_{k},\\
		&0, \text{ otherwise}
	\end{aligned}\right.
\end{equation*}
for all $k$. Observe that $(X, Y)$ and $(U, V)$ belong to the same orbit if and only if $\lvert X \cap T_{k}\rvert = \lvert U \cap T_{k}\rvert$, $\lvert Y \cap T_{k}\rvert = \lvert V \cap T_{k}\rvert$ and $\lvert X \Delta Y \cap T_{k}\rvert = \lvert U \Delta V \cap T_{k}\rvert$ for each $k$. Let $\bmi = (i_{1}, \dots, i_{m})$, $\bmj = (j_{1}, \dots, j_{m})$, and $\bmt = (t_{1}, \dots, t_{m})$ for $i_{k}, j_{k}, t_{k} \in \{0, \dots, n_{k}\}$ with $i_{k} + j_{k} - 2t_{k} \in \{0, \dots, n_{k}\}$, for all $k \in \{1, \dots, m\}$. Then the orbits of $G$ can be indexed by $(\bmi, \bmj, \bmt)$ and we may rewrite $M_{\sigma_{k}}$ as
\begin{equation*}
	(M_{\bmi, \bmj}^{\bmt})_{X, Y} = \left\{
	\begin{aligned}
		&1, \text{ if } \lvert X \cap T_{k}\rvert = i_{k}, \lvert Y \cap T_{k}\rvert = j_{k}, \lvert X \cap Y \cap T_{k}\rvert = t_{k} \text{ for all } k;\\
		&0, \text{ otherwise.}
	\end{aligned}\right.
\end{equation*}
For $\bmn = (n_{1}, \dots, n_{m})$, we denote  the algebra generated by $\{M_{\bmi, \bmj}^{\bmt}\}$ over $\mathbb{C}$, defined as above, by ${\cal A}_{\bmn}$, which is called an $m$-split Terwilliger algebra of the Hamming scheme.
	
\begin{lemma}
	${\cal A}_{\bmn}$ is isomorphic to $\bigotimes_{i = 1}^{m} {\cal A}_{n_{i}}$. 
\end{lemma}

\begin{corollary} \label{cor:blo_g}
	There is an isomorphism from ${\cal A}_{\bmn}$  to
	\begin{equation*}
		\bigoplus_{k_{1} = 0}^{\left\lfloor \frac{n_{1}}{2} \right\rfloor} \cdots \bigoplus_{k_{m} = 0}^{\left\lfloor \frac{n_{m}}{2} \right\rfloor} \mathbb{C}^{N_{\bmk} \times N_{\bmk}}, 
	\end{equation*}
	with $N_{\bmk} = \prod_{a = 1}^{m} (n_{a}-2k_{a}+1)$, that maps
	$A = \sum_{\bmi, \bmj, \bmt}x_{\bmi, \bmj}^{\bmt} M_{\bmi, \bmj}^{\bmt}$ to
	\begin{equation*}
		\bigoplus_{k_{1} = 0}^{\left\lfloor \frac{n_{1}}{2} \right\rfloor} \cdots \bigoplus_{k_{m} = 0}^{\left\lfloor \frac{n_{m}}{2} \right\rfloor} {\cal B}_{\bmk}, 
	\end{equation*}
	where
	\begin{equation*}
		\begin{aligned}
			{\cal B}_{\bmk} = \left( \sum_{\bmt} \prod_{a = 1}^{m} \binom{n_{a}-2k_{a}}{i_{a}-k_{a}}^{-\frac{1}{2}} \binom{n_{a}-2k_{a}}{j_{a}-k_{a}}^{-\frac{1}{2}}  \beta_{i_{a}, j_{a}, k_{a}}^{n_{a}, t_{a}} x_{\bmi, \bmj}^{\bmt} \right)_{(\bmi, \bmj) = (\bmk, \bmk)}^{(\bmn-\bmk, \bmn-\bmk)}
		\end{aligned} 
	\end{equation*}
\end{corollary}
	
Next we derive additional semidefinite constraints for an $(n,d)$ code $C$ from the $m$-split Terwilliger algebra. Define
\begin{equation}
	x_{\bmi, \bmj}^{\bmt} = \frac{1}{\lvert C\rvert \prod_{k = 1}^{m}\binom{n_{k}}{i_{k}-t_{k}, j_{k}-t_{k}, t_{k}}} \lambda^{\bmt}_{\bmi, \bmj},
\end{equation}
where
\begin{equation}
	\begin{aligned}
		\lambda^{\bmt}_{\bmi, \bmj} =  \lvert \{&(X, Y, Z) \in C^{3} : \lvert (X \Delta Y)\cap T_{k}\rvert = i_{k}, \lvert(X \Delta Z)\cap T_{k}\rvert = j_{k},  \\
		&\lvert((X \Delta Y)\cap (X \Delta Z))\cap T_{k}\rvert = t_{k}, \text{ for } k=1,\dots, m \} \rvert.
	\end{aligned}
\end{equation}
Now, the size of the code $C$ is 
\begin{equation*}
	\lvert C\rvert=\sum_{a = 0}^{n} \sum_{\bmi \cdot \bml = a} \prod_{k = 1}^{m} \binom{n_{k}}{i_{k}} x_{\bmi, \bmo}^{\bmo},
\end{equation*}
where $\bml$ denotes the all-one vector and  $\bmo$ denotes the zero vector. Then   the group $G$ acts on $C$ as follows:
\begin{equation*}
	(g_{1}, \dots, g_{m}) \cdot (c) = \bigcup_{k=1}^{m} g_{k}(c \cap T_{k}),
\end{equation*}
for $g_{i} \in G_{i}$, $i = 1, \dots, m$ and $c \in C$. We define the sets
\begin{align*}
	& \Pi_{\rm \bms} = \{g \in G \mid \emptyset \in g(C)\},\\
	& \Pi_{\rm \bms}' = \{g \in G \mid \emptyset \notin g(C)\},
\end{align*}
and consider the semidefinite matrices
\begin{align*}
	& R_{\rm \bms} = \frac{1}{\lvert \Pi_{\rm \bms}\rvert} \sum_{g \in \Pi_{\rm \bms}} \chi^{g(C)} (\chi^{g(C)})^{T}, \\
	& R_{\rm \bms}' = \frac{1}{\lvert \Pi_{\rm \bms}'\rvert}\sum_{g \in \Pi_{\rm \bms}'} \chi^{g(C)} (\chi^{g(C)})^{T},
\end{align*}
which are elements of ${\cal A}_{\bmn}$. In fact, we have the following proposition.
\begin{proposition}
	\label{pro:R_g}
	\begin{equation*}
		\begin{aligned}
			&R_{\rm \bms} = \sum_{\bmi, \bmj, \bmt} x_{\bmi, \bmj}^{\bmt}M_{\bmi, \bmj}^{\bmt},\\
			& 
			\begin{aligned}
				R_{\rm \bms}' =& \frac{\lvert C\rvert}{2^{n}-\lvert C\rvert}\sum_{\bmi, \bmj, \bmt} \left( x_{\bmi+\bmj-2\bmt, \bmo}^{\bmo}-x_{\bmi, \bmj}^{\bmt}\right)M_{\bmi, \bmj}^{\bmt}.
			\end{aligned}
		\end{aligned}
	\end{equation*}
\end{proposition}
	
From Proposition \ref{pro:R_g}, we can also obtain the $m$-split generalized Delsarte's inequalities~\cite{Sim95}. Consider the generalized distance distribution $\{A_{\bmi}\}$ of $C$, where
\begin{equation*}
	A_{\bmi} = \frac{1}{\lvert C\rvert}\lvert\{(a, b) \in C \times C \mid \lvert a \Delta b \cap T_{k}\rvert = i_{k} \text{ for } k=1,\dots,m\}\rvert.
\end{equation*}
\begin{corollary} 
	If $R_{\rm \bms}$ and $R_{\rm \bms}'$ are positive semidefinite, then
	\label{coro:m_gen_del}
	\begin{equation}
		\sum_{\bmi} A_{\bmi} \prod_{k = 1}^{m}K_{p_{k}}^{n_{k}}(i_{k}) \geq 0,
	\end{equation}
	for $\bmp = (p_{1}, \dots, p_{m})$ with $0 \leq p_{k} \leq n_{k}$ for all $k$.
\end{corollary}

It can be showed that the generalized Schrijver's SDP constraints with $k = m+2$ induces our positive semidefinite constraint on $R_{s}$.

\begin{proposition}
    The positive semidefinite constraints in the generalized Schrijver's SDP with $k = m+2$ implies the positive semidefiniteness of $R_{s}$ .
\end{proposition}

By the block diagonal form for ${\cal A}_{\bmn}$, we have additional semidefinite constraints. For $\bmk = (k_{1}, \dots, k_{m})$ with $k_{a} \in \{0, \dots, \lfloor \frac{n_{a}}{2} \rfloor \}$ and $a \in\{1, \dots, m\}$, the matrices
\begin{equation}
	\label{eq:gen_sdp_1}
	\begin{aligned}
		\left( \sum_{\bmt} \prod_{a = 1}^{m} \binom{n_{a}-2k_{a}}{i_{a}-k_{a}}^{-\frac{1}{2}} \binom{n_{a}-2k_{a}}{j_{a}-k_{a}}^{-\frac{1}{2}}  \beta_{i_{a}, j_{a}, k_{a}}^{n_{a}, t_{a}} x_{\bmi, \bmj}^{\bmt} \right)_{(\bmi, \bmj) = (\bmk, \bmk)}^{(\bmn-\bmk, \bmn-\bmk)},
	\end{aligned} 
\end{equation}
\begin{equation}
	\label{eq:gen_sdp_2}
	\begin{aligned}
		\left( \sum_{\bmt} \prod_{a = 1}^{m} \binom{n_{a}-2k_{a}}{i_{a}-k_{a}}^{-\frac{1}{2}} \binom{n_{a}-2k_{a}}{j_{a}-k_{a}}^{-\frac{1}{2}}  \beta_{i_{a}, j_{a}, k_{a}}^{n_{a}, t_{a}} (x_{\bmi+\bmk-2\bmt, \bmo}^{\bmo}-x_{\bmi, \bmj}^{\bmt}) \right)_{(\bmi, \bmj) = (\bmk, \bmk)}^{(\bmn-\bmk, \bmn-\bmk)}
	\end{aligned} 
\end{equation}
are semidefinite.

\begin{proposition}
    Let $C$ be a code with length $n$ and minimum distance at least $d$. Then we have the following linear constraints on $x_{\bmi, \bmj}^{\bmt}$ corresponding to $C$: for proper $\bmi, \bmj, \bmt$,
	\begin{align}
		\label{eq:gen_lin}
		\begin{array}{cl}
		    \mbox{(i) }& x_{\bmo, \bmo}^{\bmo} = 1 \\
			\mbox{(ii) }& 0 \leq x_{\bmi, \bmj}^{\bmt} \leq x_{\bmi, \bmo}^{\bmo} \\
			\mbox{(iii) }& x_{\bmi, \bmo}^{\bmo} + x_{\bmo, \bmj}^{\bmo} \leq 1 + x_{\bmi, \bmj}^{\bmt} \\
			\mbox{(iv) }& x_{\bmi, \bmj}^{\bmt} = x_{\bma, \bmb}^{\bmc} \mbox{ if } (\bmi, \bmj, \bmi+\bmj-2\bmt)  \\
			&\mbox{ is a permutation of } (\bma, \bmb, \bma+\bmb-2\bmc),  \\
			\mbox{(v) }&  x_{\bmi, \bmj}^{\bmt} = 0 \mbox{ if } \{\bmi, \bmj, \bmi+\bmj-2\bmt\} \cap \{1, \dots, d-1\} \neq \emptyset. 
		\end{array}
	\end{align}
\end{proposition}
	
To sum up, we have an SDP for $A(n, d)$  with variables $x_{\bmi, \bmj}^{\bmt} \in \mathbb{C}$:
\begin{align}
	{\rm maximize } &\sum_{a = 0}^{n} \sum_{\bmi \cdot \bml = a} \prod_{k = 1}^{m} \binom{n_{k}}{i_{k}} x_{\bmi, \bmo}^{\bmo}.  \notag\\
	{\rm subject\ to\ }       &  \mbox{ positive semidefiniteness of }(\ref{eq:gen_sdp_1}) \text{ and } (\ref{eq:gen_sdp_2}) \notag
	\\ & (\ref{eq:gen_lin})   \notag
\end{align}
This SDP is called an $m$-split SDP, which corresponds to an $m$-split Terwilliger algebra ${\cal A}_{\bmn}$.

\subsection{Underlying structure}
\label{sec:struct}
	
Herein we describe a special structure for an association scheme, which is inspired by the split method used in this paper.

\begin{definition}($m$-split property)
	Let $S = (X, \{R_{i}\}_{i=0}^{k})$ be an association scheme. We say that $S$  is \textit{$m$-split}  for some $1 \leq m \leq \lvert X\rvert$ if there exist $m$ association schemes $S_{1} = (X_{1}, \{R_{i}^{(1)}\}_{i=0}^{k_{1}}), \dots, S_{m} = (X_{m}, \{R_{i}^{(m)}\}_{i=0}^{k_{m}})$ and a collection of maps $\{f_{i}\}$ such that $f_{i}: X \longrightarrow X_{i}$ is surjective for each $i$ and for $(x, y) \in X \times X$, $(x, y) \in R_{j}$ if and only if $\sum_{i = 1}^{m} d_{i} = j$, where $(f_{i}(x), f_{i}(y)) \in R_{d_{i}}^{(i)}$ for each $i$. 
\end{definition}
		
Clearly, we have the following lemma.
\begin{lemma}
Every association scheme is $1$-split.
\end{lemma}
	
For an  $m$-split association scheme $S$, we can define an $m$-split Bose-Mesner algebra by $\bigotimes_{i=1}^{m} {\cal B}(S_{i})$ and an $m$-split Terwilliger algebra by $\bigotimes_{i=1}^{m} {\cal T}(S_{i})$, where ${\cal B}(S_{i})$ is the Bose-Mesner algebra of $S_{i}$ and ${\cal T}(S_{i})$ is the Terwilliger algebra of $S_{i}$.

In our case of binary codes, $S = {\cal H}(n, 2)$ is the Hamming scheme of length $n$ over $\mathbb{F}_{2}$ and each $S_{i}$ is ${\cal H}(n_{i}, 2)$. Clearly, $S$ is $m$-split for $1 \leq m \leq n$. We may observe that $m$-split Bose-Mesner algebras induce the $m$-split generalized Delsarte's inequalities and $m$-split Terwilliger algebras induce $m$-split triple distances SDP. Moreover, there are ascending chains for algebras.
\begin{equation}
	\label{eq:chain_1}
	{\cal B}_{1} \leq {\cal B}_{2} \leq \cdots \leq {\cal B}_{n},
\end{equation}
\begin{equation}
	\label{eq:chain_2}
	{\cal T}_{1} \leq {\cal T}_{2} \leq \cdots \leq {\cal T}_{n}.
\end{equation}
In here, ${\cal B}_{m}$ is an $m$-split Bose-Mesner algebra of $S$, and ${\cal T}_{m}$ is an $m$-split Terwilliger algebra of $S$ for $m = 1, \dots, n$. The chain (\ref{eq:chain_1}) allows us to combine an $l$-split generalized Delsarte's inequalities to an $m$-split generalized Delsarte's inequalities for $l \leq m$ by merging some partitions as the one. On the other hand, the chain (\ref{eq:chain_2}) allows us to add constraints in an $l$-split triple distances SDP to an $m$-split triple distances SDP for some $l \leq m$ without increasing the number of variables.

\section{Conclusion}
\label{sec:conslus}
In conclusion, we have derived generalized Schrijver semidefinite constraints by considering the split of Terwilliger algebra. Our split semidefinite constraints are a natural generalization of Schrijver's constraints since they also implied the generalized Delsarte inequalities. 
	
By implementing a $2$-split  SDP, we have improved the upper bounds for $A(18, 4)$ and $A(19, 4)$. 	The MATLAB programs of the SDPs in this paper can be found at: 
\begin{center}
	\url{https://github.com/PinChiehTseng/Split_SDP_solution}
\end{center}
\noindent 	
We have confirmed that $A(18,4)\leq 6551$ by showing that an upper bound on the numerical error is small enough. As for $A(19,4)$, we obtain $A(19,4)\leq 13087.5$  using $n_{1} = 9$.  The currently best known bound for $A(19,4)$ is $13104$.  As for the error estimate,  we obtained an   upper bound on the numerical error as large as   $215.7376$. Since this estimate is over pessimistic, it provides no information about out result. Since the  solver normally returned  without any warning, we believed this figure is correct. More accurate solvers could be considered. 

Our split approach can be extended to other related problems. For example, one may consider an arbitrary finite field or $k$-distance SDP for arbitrary $k$. Notice that the method described in \cite{GMS12} with $\lvert S \rvert = 2$ has been applied to improve upper bounds for constant weight codes. Moreover, the algebra considered in \cite{Pol19_2} is of the form $\bigotimes_{i} \mathcal{A}_{n_{i}}$ with $\sum_{i}n_{i} = n$. The constraint $R_{s}$ for constant weight codes has been studied. We can have a similar application by adding the constraint $R'_{s}$, which potentially opens a way to strengthen the upper bounds for constant weight codes. As for the upper bounds on the size of a set  with few distances or intersecting families of subsets \cite{MN11, BM11}, we may directly apply our method to those problems. Moreover, as Schrijver's SDP has been extended to the maximum size problem of a code in the fold of $n$-cube by Hou et al. \cite{HHGY20}, we might have a similar extension. The idea of isometry groups might be applied to spherical codes as well. 

As an example of association schemes, our method corresponds to a special case of the $m$-split property of the Hamming scheme. Thus, our method may be extended to other association schemes, which share this $m$-split property. It is an interesting research direction.

Finally, it has been shown that certain polynomial symmetric properties can be exploited to obtain additional matrix inequalities and hence improve the  semidefinite programming bounds for the kissing number problem~\cite{caluza2018improving}. It is unknown whether a similar ideas  could be applied to the case of binary codes to obtain better bounds for $A(n,d)$.  
	
\section*{ACKNOWLEDGEMENT}
	
We would like to thank Dion Gijswijt, Alexander Schrijver, and Hajime Tanaka for helpful discussions. We would also like to thank the anonymous referees for their valuable comments.

PCT and CYL were supported by the Ministry of Science and Technology (MOST) in Taiwan under Grant MOST110-2628-E-A49-007. WHY was supported by MOST under Grant109-2628-M-008-002-MY4.

\newcommand{\etalchar}[1]{$^{#1}$}

\end{document}